\documentclass[reprint,amsmath,amssymb,aps,onecolumn,pra]{revtex4-2}

\usepackage{dcolumn}% Align table columns on decimal point
\usepackage{bm}% bold math

\usepackage{relsize}
\usepackage[utf8]{inputenc}
\usepackage[english]{babel}
\usepackage[T1]{fontenc}
\usepackage{hyperref}
\usepackage{graphicx}

\usepackage{amsthm}
\usepackage{amsmath} 
\usepackage{amssymb}
 
\newtheorem{theorem}{Theorem}
\newtheorem{corollary}[theorem]{Corollary}
\newtheorem{lemma}[theorem]{Lemma}
\newtheorem{definition}[theorem]{Definition}

\begin{document}

\preprint{APS/123-QED}

\title{Local-dimension-invariant Calderbank-Shor-Steane Codes with an Improved Distance Promise}% Force line breaks with \\
\author{Arun J. Moorthy}
\affiliation{BASIS Scottsdale, Scottsdale, Arizona, 85259, United States of America}

\author{Lane G. Gunderman}
 \email{lgunderman@uwaterloo.ca}
\affiliation{%
 The Institute for Quantum Computing, University of Waterloo, Waterloo, Ontario, N2L 3G1, Canada
}%
\affiliation{Department of Physics and Astronomy, University of Waterloo, Waterloo, Ontario, N2L 3G1, Canada}

\date{\today}

\begin{abstract}
Quantum computers will need effective error-correcting codes. Current quantum processors require precise control of each particle, so having fewer particles to control might be beneficial. Although traditionally quantum computers are considered as using qubits (2-level systems), qudits (systems with more than 2-levels) are appealing since they can have an equivalent computational space using fewer particles, meaning fewer particles need to be controlled. In this work we prove how to construct codes with parameters $[[2^N,2^N-1-2N,\geq 3]]_q$ for any choice of prime $q$ and natural number $N$. This is accomplished using the technique of local-dimension-invariant (LDI) codes. Generally LDI codes have the drawback of needing large local-dimensions to ensure the distance is at least preserved, and so this work also reduces this requirement by utilizing the structure of CSS codes, allowing for the aforementioned code family to be imported for any local-dimension choice.
\end{abstract}

\maketitle

Mistakes are bound to happen. Quantum computers, much like their classical counterparts, must overcome and correct their errors in order to perform reliable computations. \if{false}Current quantum computers require very expensive, carefully calibrated controls, such as cryogenic environments \cite{devoret2004superconducting} or many lasers \cite{kielpinski2002architecture}.\fi Building a modest sized quantum computer currently is a significant challenge. If, however, qudits, particles with $q$-levels (or local-dimension $q$), are used in the place of qubits, particles with just $2$-levels, we can drastically reduce the number of \textit{particles} that need protection while still having a similarly sized computational space \cite{wang2020qudits}. This relationship can also be captured in the formula for the number of logical codewords for a qudit stabilizer code, $q^k$, which states that the number of codewords grows exponentially with the local-dimension $q$ as the base, so increasing $q$ increases the amount of information protected per particle \cite{ketkar2006nonbinary}. As of this time there are currently some qudit quantum computers which are in early development \cite{qudition,quditlight,molecules,yurtalan2020implementation}. We aim here to provide more options for error-correcting codes for qudit systems, as well as show that the local-dimension-invariant framework can provide at least one family with the distance always being promised.% While most quantum computers being built are based on qubits, having fewer particles to manipulate could be advantageous to reduce phononic noise and more rapid warming up in cryogenic environments due to larger collections of particles.

A common tool used for correcting errors in the quantum case are stabilizer codes, the quantum analog of classical linear codes \cite{gottesman1997stabilizer}. For these codes, there are four parameters that are traditionally used to specify the effectiveness of the code, written as $[[n,k,d]]_q$. $n$ specifies the number of physical qudits, or particles, that the code uses to protect $k$ logical qudits. The level of protection is specified by $d$, which is the distance of the code. This provides a limit on the rate of errors we are allowed in our system while still ensuring that the information is preserved \cite{nielsen2002quantum}. Lastly, $q$ is the number of levels the system has which are being used, also known as the local-dimension, and is assumed throughout this work to be some prime number. The choice of the local-dimension always being a prime number is primarily to ensure unique multiplicative inverses--it is likely that this restriction can be at least somewhat loosened. An ideal quantum code would have both a high number of protected qudits, $k$, and distance, $d$. Unfortunately, the Generalized Quantum Hamming Bound, given by: 
\begin{equation}
    \sum_{j=0}^{\lfloor \frac{d-1}{2}\rfloor} {n\choose j} (q^2-1)^j\leq q^{n-k},
\end{equation}
shows that there’s a trade-off between $k$ and $d$--and some dependence on $n$ and $q$--at least in the case of non-degenerate codes \cite{ketkar2006nonbinary}. In general, the higher the distance is, the fewer the number of logical qudits that are protected. While this bound does not always hold in the case of degenerate codes, for some classes of codes this bound is still true \cite{sarvepalli2010degenerate,gottesman1997stabilizer}.

%Quantum error-correcting codes formally use unitary operators, however, due to the stabilizer formalism they can instead be expressed as matrices where all entries are in $\mathbb{Z}_q=\{0,1,..., q-1\}$, and so are similar to classical linear codes in the following ways: elementary matrix operations preserve the distance as does swapping particle/bit numbers. Essentially, using stabilizer codes is really effective because it allows us to deal with simpler linear algebra rather than more complex and elaborate math. 

\if{false}Quantum error-correction effectively began when Peter Shor proposed his 9-qubit code from the composition of two 3-qubit codes; where the inner code protects against bit-flips and the outer code protects against phase-flips, and so creating a somewhat error-resistant single logical qubit \cite{shor1995scheme}. Following this Gottesman, as well as other papers, extended these examples into the so-called stabilizer formalism and produced the Calderbank-Shor-Steane (CSS) theorem for importing some classical codes as methods for protecting quantum information \cite{gottesman1997stabilizer,calderbank1996good,steane1996error}. Further work extended this mapping to classical codes from the binary (qubit) case to the nonbinary (qudit) case \cite{gottesman1998fault, ketkar2006nonbinary}. This allowed for the definition and proofs for many qudit codes. Although this defined many families of codes, many of these code families have challenging relationships required between their parameters, additionally these families have sets of parameters without a corresponding code, and so loosening these constraints could be useful \cite{ketkar2006nonbinary}. One avenue for this would be to take heavier inspiration from classical binary codes, where the majority of classical codes exist, and somehow apply them for qudit systems as well.\fi While most codes have been designed for a given local-dimension value, Chau showed that both the 5-qubit \cite{chau1997five} and 9-qubit \cite{chau1997correcting} codes could be transformed into qudit codes with the same parameters and minimal modifications to the code. These examples were extended into a framework allowing all qudit codes with $q$ levels to be transformed into valid qudit codes over $p$ levels, however, it only showed that the distance could be preserved once there are sufficiently many levels; beyond some critical value $p^*$, which is a function of $n$, $k$, $d$, and $q$ \cite{gunderman2020local}. These works paved the path to using codes over various local-dimensions, however, this latter work left a large gap between $q$ (the initial dimension) and $p^*$ where statements about the distance cannot be made without manually determining the distance. In this report we consider the subclass of stabilizer codes known as Calderbank-Shor-Steane (CSS) codes which allows for a significantly reduced local-dimension requirement for this class in order for the distance to be preserved. Following this we consider one family of a qubit CSS code and show that at least one local-dimension-invariant (LDI) representation can be used with the same parameters regardless of the local-dimension.

\section{Background}

In this section we provide a brief background on qudit stabilizer operators. For a more full review see \cite{ketkar2006nonbinary}. In order to allow for easier discussion of qubits and qudits, the general term "particle" or "register" is used. Throughout this work we set $\mathbb{Z}_q=\{0,1,\ldots, q-1\}$, where $q$ is a prime number.

For qudits the standard Pauli matrices are replaced by the generalized Paulis $X_q$ and $Z_q$, which act as follows:
\begin{equation}
    X_q|j\rangle=|j\oplus 1\rangle,\quad Z_q|j\rangle =\omega^j |j\rangle,\quad \omega=e^{2\pi i/q},\quad j\in \mathbb{Z}_q,
\end{equation}
where $\oplus$ is addition mod $q$. For notational clarity, we will drop the subscript on these operators. These matrices, like their qubit counterparts, form a complete basis for the set of Unitary operators in $SU(\mathbb{C}^q)$, which means that any error can be decomposed as a linear combination of powers of these operators \cite{ashikhmin2001nonbinary}. These operators also form a group, which over a single particle is indicated by $\mathbb{P}_q$. This group structure is preserved over tensor products of these operators, so a generalized Pauli acting on $n$ particles will be in the group $\mathbb{P}_q^n$. 

The generalized Pauli operators follow the following commutation relation: $X^a Z^b=\omega^{-ab} Z^b X^a$. A stabilizer code with $n-k$ generators is equivalent to a commuting subgroup of size $q^{n-k}$ of $\mathbb{P}_q^n$ which does not contain a nontrivial multiple of the identity operator. %This stabilizer code is specified by a set of parameters $[[n,k,d]]_q$ where: $n$ is the number of physical particles in the code, $k$ is the number of logical particles in the code, $d$ is the distance of the code, and $q$ is the local-dimension.

Upon quotienting out the leading scalar, the Pauli operators can be written as vectors using the symplectic representation. As we will be varying the local-dimension we use the slightly more flexible $\phi$ representation which also permits specification of the local-dimension. This mapping carries $\phi_q: \mathbb{P}_q^n\mapsto \mathbb{Z}_q^{2n}$ and is defined by:
\begin{equation}
    \phi_q \left(\bigotimes_{t=1}^n X^{a_t}Z^{b_t}\right)=\left(\bigoplus_{t=1}^n a_t\right) \bigoplus \left(\bigoplus_{t=1}^n b_t\right),
\end{equation}
where in the above $\bigoplus$ is a direct sum symbol. This mapping follows the law of composition of $\phi_q(p_1\circ p_2)=\phi_q(p_1)\oplus\phi_q(p_2)$, where $\oplus$ is entry-wise addition $\mod q$ and $p_1$ and $p_2$ are two generalized Pauli operators. We will write a vertical bar between the vector for the $X$ powers and the vector for the $Z$ powers mostly for ease of reading. The special case of $\phi_\infty$ which follows the same relations, but for this representation one does not take modulo any value but carries computations over the integers, was proven in \cite{gunderman2020local}. As an example of the difference, consider the following:
\begin{equation}
    \phi_2(X\otimes Z^{-1} \otimes I \otimes XZ)=(1\ 0\ 0\ 1\ |\ 0\ 1\ 0\ 1),\quad \phi_\infty(X\otimes Z^{-1} \otimes I \otimes XZ)=(1\ 0\ 0\ 1\ |\ 0\ -1\ 0\ 1).
\end{equation}

In the $\phi$ representation the commutator of two generalized Paulis, $p_1$ and $p_2$, is written and computed as $\phi(p_1)\odot \phi(p_2)=\vec{a}^{(1)}\cdot\vec{b}^{(2)}-\vec{b}^{(1)}\cdot\vec{a}^{(2)}$. This is not formally a commutator, but when this is zero, or zero modulo the local-dimension, the two operators commute, while otherwise it is a measure of number of times an $X$ operator passed a $Z$ operator without a corresponding $Z$ operator passing an $X$ operator.

The distance of a stabilizer code is a measure of its error-protecting capabilities. The standard choice of the depolarizing channel is considered here, meaning that the distance of the code is measured in terms of the Pauli weight of the error, given by the number of non-identity operators in the Pauli error. To aid in determining which error might have occurred the \textit{syndrome} values are computed by finding the commutator of the error with each of the generators for the stabilizer code.

A stabilizer code is written in the $\phi$ representation as a matrix whose rows are a set of $n-k$ generators for the subgroup. There are some operations that we may perform which must preserve the parameters of the code, these include, in the $\phi$ representation:
\begin{itemize}
    \item Row swaps, corresponding to relabelling the generators.
    \item Swapping columns $i$ and $i+n$ with $j$ and $j+n$, corresponding to relabelling particles.
    \item Multiplying a row by any number in $\{1,\ldots, q-1\}$, corresponding to composing that generator with itself.
    \item Adding row $i$ to row $j$, corresponding to composing the operators.
    \item Swapping column $i$ with $-1$ times column $i+n$, corresponding to a discrete-fourier transform (DFT) on particle $i$; the qudit analog of the Hadamard gate.
\end{itemize}
We neglect the phase gate, $\sqrt{Z}$, since we do not use it here, but it would also preserve the parameters of the code. Note though that the SUM gate (qudit CNOT gate) will not usually preserve the distance of the code so we do not allow ourselves to perform that operation on our codes \cite{gottesman1998fault}. Now that the tools have been presented we proceed to our results.

\if{false}
\section{Random Qudit Code Generation}

%Reasoning for why, brief algorithm and new codes, then transition to LDI for "quantum Moore's law".

While many qudit code families are known, these families have many gaps in their possible parameters \cite{ketkar2006nonbinary}. While asymptotically this is not a concern, for small qudit systems we will likely not be able to utilize the full capability of the designed qudit computer without custom designed codes. In the case of qubits there is an enumeration of the best code parameters \cite{Grassl:codetables}. One way to generate such customized qudit codes is to randomly generate them.

In this case we take in a set of desired parameters: $n$, $k$, $d$, and $q$. Then, so long as these parameters satisfy the generalized quantum Hamming bound, we begin by generating a single $2n$ length vector of elements from $\mathbb{Z}_q$. Following this we randomly generate another $2n$ length vector of elements from $\mathbb{Z}_q$, until we generate one which commutes with the prior vector. We then repeat this, requiring each successive vector to commute with all prior ones, until we have obtained $n-k$ linearly independent rows. This only satisfies the parameters $n$, $k$, and $q$. Finally the distance is checked and if it is not the desired value, then the process is repeated until such is found--or terminate as it either does not exist or is very hard to find. Formally the requirement that the parameters satisfy the generalized quantum Hamming bound is not per se required, as the code generated could be a degenerate code.% as it is still not known whether degenerate codes must also satisfy this bound, although all known codes so far do satisfy this bound \cite{ketkar2006nonbinary}.

Using this procedure we have randomly generated the following improved codes, none of which can exist as qubit codes due to the generalized quantum Hamming bound:
\begin{equation}
[[6,2,3]]_3=\begin{bmatrix}
2 & 2 & 2 & 1 & 0 & 1 & | & 1 & 1 & 0 & 2 & 0 & 2\\
1 & 1 & 2 & 0 & 2 & 0 & | & 2 & 0 & 0 & 2 & 2 & 2\\
0 & 2 & 2 & 0 & 1 & 2 & | & 2 & 1 & 2 & 2 & 1 & 0\\
0 & 1 & 1 & 1 & 0 & 2 & | & 0 & 0 & 2 & 0 & 1 & 0
\end{bmatrix}
\end{equation}

\begin{equation}
[[7,3,3]]_3=\begin{bmatrix}
1 & 2 & 1 & 2 & 0 & 0 & 0 & | & 0 & 0 & 1 & 2 & 1 & 2 & 2\\
2 & 1 & 2 & 0 & 1 & 1 & 0 & | & 1 & 0 & 1 & 0 & 2 & 1 & 2\\
2 & 0 & 1 & 2 & 1 & 0 & 1 & | & 0 & 2 & 0 & 2 & 2 & 0 & 2\\
1 & 1 & 0 & 1 & 2 & 2 & 2 & | & 0 & 0 & 0 & 0 & 0 & 2 & 0
\end{bmatrix}
\end{equation}

\begin{equation}
[[8,4,3]]_3=\begin{bmatrix}
1 & 1 & 2 & 0 & 0 & 1 & 1 & 1 & | & 2 & 1 & 0 & 2 & 0 & 0 & 0 & 0\\
1 & 0 & 1 & 1 & 0 & 1 & 2 & 0 & | & 1 & 1 & 2 & 1 & 0 & 2 & 1 & 1\\
0 & 1 & 2 & 1 & 2 & 1 & 0 & 2 & | & 0 & 0 & 1 & 1 & 2 & 0 & 1 & 0\\
1 & 0 & 0 & 1 & 0 & 2 & 1 & 1 & | & 1 & 0 & 2 & 2 & 2 & 0 & 0 & 2
\end{bmatrix}
\end{equation}

\begin{equation}
[[10,5,3]]_3=\begin{bmatrix}
2 & 0 & 2 & 2 & 0 & 2 & 2 & 0 & 0 & 0 & | & 0 & 1 & 2 & 0 & 1 & 1 & 1 & 0 & 2 & 2\\
2 & 0 & 0 & 0 & 2 & 2 & 0 & 1 & 1 & 1 & | & 1 & 1 & 0 & 2 & 2 & 0 & 1 & 1 & 1 & 2\\
2 & 1 & 0 & 0 & 2 & 0 & 1 & 1 & 2 & 2 & | & 2 & 0 & 1 & 2 & 1 & 2 & 2 & 1 & 2 & 2\\
0 & 0 & 2 & 0 & 1 & 2 & 0 & 1 & 2 & 1 & | & 0 & 2 & 2 & 0 & 1 & 2 & 1 & 2 & 1 & 1\\
2 & 2 & 0 & 0 & 2 & 2 & 0 & 0 & 0 & 0 & | & 2 & 0 & 1 & 1 & 0 & 2 & 0 & 0 & 1 & 2
\end{bmatrix}
\end{equation}

%%%%GENERATE A COUPLE MORE EXAMPLES? MAYBE? PERHAPS q=5 and n=9 and d>3?

All of these codes protect few enough qudits that they are of more immediate use. Aside from the parameters considered, other properties of codes can be of particular use for fault-tolerance such as: transversality, cyclicity, and low-weight generators. We focus on this last property here. Having lower Pauli weight generators means fewer gate estimate errors can build up to become a problem. For instance, if each gate has some error of rotation of $\epsilon$, then having $t$ non-trivial operations performed allows an error now of size roughly $t\epsilon$, so having a smaller $t$ reduces the effects of these errors. Having the lowest weight generators possible is of use in procedures such as flag fault-tolerance \cite{chao2018quantum,chao2020flag}. For this reason, the program is also able to compute the minimum choice of the maximal weight generator required to specify the code: 
\begin{equation}\label{wts}
    wt_p(\mathcal{C}):=\min_S \max_{s_i\in S} wt_p(s_i),
\end{equation}
where $wt_p(s_i)$ is the number of non-identity operators in the generator $s_i$. Note that this weight does not change under the operations we are allowed. For the codes listed above the minimal Pauli weight of the generators are given by: $5$, $6$, $7$, and $8$, respectively. For completeness the generators with this lowest maximal weight are reported in the appendix. As another example, consider the qubit quantum Hamming code, $\mathcal{C}_N$, indexed by $N$, which has $wt_p(\mathcal{C}_N)=2^{N-1}$.

While these codes only increase the number of protected \textit{particles} by one or two, the number of \textit{computational states}, the different states that would be used when performing operations, is increased more drastically. This difference becomes more clear upon considering these results within the recent \textit{local-dimension-invariant} (LDI) results \cite{gunderman2020local}. LDI codes have all generators for the code not just satisfying $\phi_q(s_i)\odot \phi_q(s_j)=0\mod q$, but the stronger condition of $\phi_\infty(s_i)\odot\phi_\infty(s_j)=0$. LDI codes can be used over any local-dimension choice, although the distance is only promised to remain the same under upon selecting a new local-dimension beyond some cutoff value for the local-dimension. A prescriptive method that allows for all stabilizer codes to be put into LDI form is provided in that work, which we will revisit in the next section.

Since the parameters for all of the codes we found cannot exist for qubit codes, due to the generalized quantum Hamming bound, the distance only has the possibility of being preserved for all larger local-dimensions. Below are plots of $q^k$ for each of the four codes listed above comparing the number of computational states of those to those of the best qubit codes (when used in LDI form). The plots illustrate that the improvement in the number of protected particles is equivalent to an increase in the \textit{power} of the growth of the number of computational states.

%We plot the number of computational states of these improved randomly generated codes to those of the best possible qubit code, where we have assumed that we may preserve distance of the LDI form of the code. In all of the plots this boils down to changing the power of the local-dimension $q$ as the number of computational states (bases for performing computations) is given by $q^{n-k}$, which means that the second entry in each of these parameter sets indicates the power of the growth of the number of computational states.

\begin{figure}[htb]
\includegraphics[width=\textwidth]{Arun Moorthy - Quantum Error Correction AzSEF.pptx (5) (1).png}
\fcaption{These plots compare the number of computational states of the new randomly generated code and the original qubit code upon being put into a distance preserving LDI form. Upper left is [[6,1,3]] to [[6,2,3]]; upper right is [[7,1,3]] to [[7,3,3]]; lower left is [[8,3,3]]  to [[8,4,3]]; lower right is [[10,4,3]] to [[10,5,3]]. In all cases these codes only exist for qudits.}
\centering
\end{figure}

With these randomly generated codes we have reported new codes seemingly not priorly known as well as codes whose number of computational states grows far more rapidly with the local-dimension than those already known for qubits with those particular values of $n$ and $d=3$. While these examples are themselves of some interest, the utility of this script is even greater as any user may generate custom qudit codes with parameters they desire, so long as they are possible.
\fi

\section{Local-dimension-invariant CSS Codes}

%While there are already many known qudit quantum error-correcting codes, many of these codes have extremely challenging restrictions between their parameters: $n,k,d,q$. Being able to bring some of these codes over from one value of $q$ to some other local-dimension value $p$ will provide many more code options. Without these extra options we may design a physical qudit platform, but find that there are no codes utilizing the number of particles we have, and we must instead discard some of the qudits, which is a costly sacrifice. We utilize the new technique of local-dimension-invariant (LDI) codes to achieve this \cite{gunderman2020local}.

\if{false}
The simplest example of a LDI code, generated from an already known code is the following:
\begin{equation}
    \xi=\begin{bmatrix}
    1 & 1 & | & 0 & 0\\
    0 & 0 & | & 1 & 1
    \end{bmatrix}
\end{equation}
This code has $n=2$ qubits, $k=2$ generators, so $n-k=0$, which means that a single state is protected from some errors, and the distance of the code gives $d=2$. Now, if we compute the commutator of the rows (which are the generators of the code), then we obtain:
\begin{equation}
    (1\cdot 1 +1\cdot 1)-(0\cdot 0+0\cdot 0)=2
\end{equation}
This commutator value is congruent to 0 when we have qubits as $2\mod 2=0$, but if we tried to use this as a qutrit code, these two generators no longer commute as $2\mod 3=2$ and so it is not a valid stabilizer code. If instead we used the code generated by:
\begin{equation}
    \xi'=\begin{bmatrix}
    1 & -1 & | & 0 & 0\\
    0 & 0 & | & 1 & 1
    \end{bmatrix}
\end{equation}
The commutator is now exactly 0, and so no matter what choice of local-dimension $q$ we pick, the commutator will always be congruent to 0, meaning that these form a valid stabilizer code for any choice of local-dimension $q$. This ability to change the local-dimension is what makes a code local-dimension-invariant. This can always be accomplished for all stabilizer codes.\fi

We now proceed to a new result related to local-dimension-invariant (LDI) codes in the case of Calderbank-Shor-Steane (CSS) codes. CSS codes have a set of independent generators only using $X$ operators or $Z$ operators for each generator. While not the only way to generate local-dimension-invariant codes, the prescriptive method from \cite{gunderman2020local} provides one such method. Any stabilizer code $\mathcal{S}$ can be put into canonical form, which is given by:
\begin{equation}
    \mathcal{S}=\begin{bmatrix}
    I_{n-k} & X_2 & | & Z_1 & Z_2
    \end{bmatrix}.
\end{equation}
The first $(n-k)\times (n-k)$ block is the $n-k$ dimensional identity matrix, while $X_2$ is a $(n-k)\times k$ matrix, $Z_1$ is a $(n-k)\times (n-k)$ matrix, and $Z_2$ is a $(n-k)\times k$ matrix, where the last three matrices' entries are determined from the operations performed to make the first $(n-k)\times (n-k)$ block the identity matrix. Then the prescriptive method is to then transform it into:
\begin{equation}
    \mathcal{S}'=\begin{bmatrix}
    I_{n-k} & X_2 & | & Z_1+L\ & Z_2
    \end{bmatrix}
\end{equation}
where $L$ is a $(n-k)\times (n-k)$ matrix whose nonzero entries are given by: $L_{ij}=\phi_\infty (s_i)\odot \phi_\infty(s_j)$, when $i>j$. $\mathcal{S}'$ satisfies $\phi_\infty(s_i')\odot \phi_\infty(s_j')=0$ and $\mathcal{S}'=\mathcal{S}\mod q$. This merely generates a set of commuting operators, but ensuring the distance of the code is a far more challenging task.

Figure 1 illustrates the initial distance promises proven for arbitrary non-degenerate codes, as well as the improvements shown in this work. For local-dimensions $p$ with $q<p<p^*$, the distance of the code is uncertain. \if{false}For local-dimensions $p$ with $p<p^{**}$ the distance of the code when put into LDI form must decrease, while for $p^{**}<p<p^{*}$, $p\neq q$, the distance of the code is uncertain.\fi For this uncertainty region the distance could be lower, it could be the same, or it could even be higher, however, it must be determined manually. It is conjectured that for $q\leq p$ the distance can be at least preserved, however, this is only known to be possible for a few codes, such as the $5$, $7$, and $9$ particle codes \cite{chau1997five,gunderman2020local,chau1997correcting}. This work adds to this collection an infinite family. \if{false}, however, until such is shown, any local-dimension choice in the uncertain region will need its distance checked.\fi While the distance is promised to be at least the same for $p^*<p$, the caveat is that the current bound for $p^*$ for an arbitrary non-degenerate stabilizer code is very large: $p^*=B^{2(d-1)}(2(d-1))^{d-1}$ with $B$ being the largest entry in the LDI form, which is bounded by $B\leq (2+k(q-1))(q-1)$ \cite{gunderman2020local}. This value is typically large, so reductions to it is crucial to make the technique of more practical use. We show next that in the class of CSS codes we may reduce this cutoff value roughly quadratically: $p_{CSS}^*\approx \sqrt{p^*}$.

\if{false}
\begin{figure}[htb]
\includegraphics[width=\textwidth]{Arun Moorthy - Quantum Error Correction AzSEF.pptx (7) (1).png}
\caption{The green shaded region is the set of local-dimension values for which the distance of the code is uncertain.}
\centering
\end{figure}
\fi
%Our program provides two ways in which it can help provide insights into the green shaded region where the distance of the code is uncertain. These insights are: 1) being able to check the distance at a given choice of $p$ for this prescriptive technique, and 2) being able to greatly reduce the value for $p^*$ for a given code by explicitly constructing the local-dimension-invariant form and finding its maximal entry.

%In the first option we begin with some qudit code over $q$ levels, then apply the prescriptive technique to generate an LDI form for the code. Once this is done, we select a new local-dimension $p$ and compute the distance of the code. The following table contains our results for a few example cases:

\if{false}
\begin{figure}[htb]
\includegraphics[width=0.7\textwidth]{Arun Moorthy - Quantum Error Correction AzSEF.pptx (2).png}
\caption{This table provides a list of example results from probing the distance at some new chosen local-dimension $p$.}
\centering
\end{figure}
\fi

%The second option takes an initial code over $q$ levels, applies the prescriptive technique, and using the explicitly constructed code computes a better bound on the largest entry in the LDI form. The expression for $p^*$ is given by $B^{2(d-1)}(2(d-1))^{(d-1)}$, where $B$ is the largest entry in the LDI form. Without constructing the code we must take the worst case value possible for $B$ which is $(2+(n-k)(q-1))(q-1)$, but as our data finds, this value of $B$ is often extremely loose. Indeed, we consider the quantum Hamming family \cite{gottesman1996class} and find that by using our program we can determine a better value for $B$ such that we remove roughly $99.6\%$ of the uncertain region for this family. This means that only $0.4\%$ of the original set of $p$ values must have their distance determined computationally or through other means, which greatly improves the viability of using these codes as qudit error-correcting codes as well. While this is not as immediately helpful, this technique, especially if refined, could be invaluable for future qudit platform implementations. The ratios of the improved $p^*$ value from constructing the code compared to not constructing the code is also plotted below.

\if{false}
\begin{figure}[htb]
\includegraphics[width=\textwidth]{Arun Moorthy - Quantum Error Correction AzSEF.pptx (3) (1).png}
\fcaption{This image illustrates the effect of finding a reduced $p^*$ value in the qubit quantum Hamming family by obtaining a better bound on the maximal entry. The parameters of this code family are originally $[[2^N-1,2^N-2N-1,3]]_2$. Knowing this improved bound cuts down the size of the number of local-dimension values that must have their distance manually determined.}
\centering
\end{figure}
\fi
\if{false}
\begin{figure}[htb]
\includegraphics[width=\textwidth]{Arun Moorthy - Quantum Error Correction AzSEF.pptx (1) (1) (1).png}
\caption{This image illustrates the effect of finding a reduced $p^*$ value by knowing the maximal entry value. Knowing this cuts down the number of local-dimension values that must have their distance manually determined.}
\centering
\end{figure}
\fi

\begin{figure}[htb]
\includegraphics[width=\textwidth]{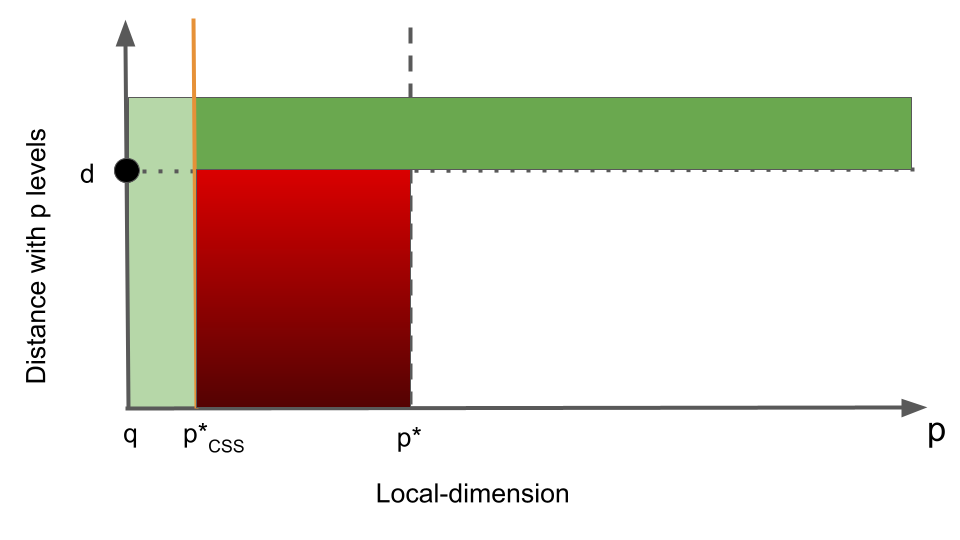}
\caption{This schematic image illustrates the effect of finding a reduced expression for the cutoff value for the distance promise. The light green area on the left is the set of local-dimension values which must have their distance manually checked, the red region is the set of local-dimension values which must have their distance manually checked for a CSS code if the general $p^*$ expression is used, whereas when $p_{CSS}^*$ is used that region automatically has the distance promised.}
\centering
\end{figure}

%Through implementing this technique and analyzing it, we have greatly increased the viability of the LDI forms for quantum codes, and hence expanded the set of code options for a party building a qudit quantum computer.

%\subsection{LDI New Results}

\begin{theorem} \label{css}
For all primes $p$, $p_{CSS}^*<p$, with $p_{CSS}^*=B^{d-1}(d-1)^{(d-1)/2}$, the distance of a local-dimension-invariant non-degenerate CSS code with parameters $[[n,k,d]]_q$ used over $p$ levels with parameters $[[n,k,d']]_p$, has $d'\geq d$.
\end{theorem}

In order to prove Theorem \ref{css} we begin by recalling a pair of definitions used to categorize undetectable errors \cite{gunderman2020local}. Undetectable errors are those Pauli operators whose syndrome values for all generators are congruent to zero, but are not themselves part of the subgroup formed by the generators of the code. The following definitions form a complete description of all undetectable errors:

\begin{definition}
An unavoidable error is an error that commutes with all stabilizers and produces the $\vec{0}$ syndrome over the integers.
\end{definition}

These were dubbed such as no matter the local-dimension these errors will always be unable to be detected, this contrasts with the other category:

\begin{definition}
An artifact error is an error that commutes with all stabilizers but produces at least one syndrome that is only zero modulo the base.
\end{definition}

For errors within this category the fact that these errors are not able to be detected is a feature of the local-dimension. If this code, and this error, were considered over a different local-dimension then the error might be able to be detected, meaning the inability to notice this error is merely an artifact of the local-dimension of the system.

With these definitions, the proof of Theorem \ref{css} follows using largely the same argument as Theorem 16 from \cite{gunderman2020local}.

\begin{proof}We begin by recalling that Example 15 from \cite{gunderman2020local} showed that LDI codes, using the given method, preserve CSS structure. CSS codes take the structure:
\begin{equation}
    \begin{bmatrix}
    \mathcal{X} & | & 0\\
    0 & | & \mathcal{Z}
    \end{bmatrix}.
\end{equation}
The distance preservation property may then be considered separately within the two blocks $\mathcal{X}$ and $\mathcal{Z}$. Each block only has nonzero valued syndrome values for $Z$ and $X$ terms in a Pauli respectively. Given this, the Pauli weight for any error activating syndromes in just the $\mathcal{X}$ portion is the same as the Hamming weight for those operators. This is likewise the case for the $\mathcal{Z}$ portion. This then means that in order to have an undetectable error it must be undetectable within the $\mathcal{X}$ portion and within the $\mathcal{Z}$ portion.

Next, undetectable errors correspond to nontrivial kernels of matrix minors. Consider building up the errors composed solely of $Z$ operators by their Hamming weight $w$. The kernel is nontrivial if the $w\times w$ matrix minor corresponding to this Hamming string has a determinant that is congruent to zero modulo the local-dimension. To avoid tracking the locations of the Hamming weights we allow ourselves to arbitrarily permute the entries within each portion of the CSS code. The determinant for a minor can be zero in two different ways:
\begin{itemize}
    \item If the determinant is 0 over the integers then this error is either an unavoidable error or an error whose existence did not occur due to the choice of the local-dimension.
    \item If the determinant is not 0 over the integers, but is merely some multiple of the local-dimension, then this corresponds to an artifact error.
\end{itemize}
This means that we only have to worry about this second case lowering the distance upon changing the local-dimension. If the determinant of this minor can be guaranteed to be smaller than the local-dimension then we are promised that the distance will at least remain the same. We may bound the determinant of a $w\times w$ matrix minor using Hadamard's inequality, which we evaluate at $d-1$ since we only need to ensure no artifact errors are induced prior to weight $d$. This then provides:
\begin{equation}
    p_{CSS}^*=B^{d-1}(d-1)^{(d-1)/2}.
\end{equation}
Then so as long as the local-dimension is larger than this no $Z$ error will result in a distance lower than $d$. This same argument can be made for $X$ errors, which completes the proof.
%This means that we only need consider permutations of all matrix minors of size at most $d\times d$. Due to this we only need use Hadamard's inequality for matrices up to this size, providing $p^*=B^{d-1}(d-1)^{(d-1)/2}$.
\end{proof}

In essence this is based on needing only to perform the prior distance promise proof over each of the two blocks in these codes and only needing to worry about the Pauli weight of the error in each block being the same as the Hamming weight. This is slightly better than a quadratic improvement over the general $p^*$ bound. This is a step toward decreasing the value of $p^*$ significantly enough to make this method of practical use. Let us consider the utility of this improvement for the qubit Hamming family, with parameters $[[2^N-1,2^N-1-2N,3]]_2$. The first member of this family, $H_3$, is known as the Steane code with parameters $[[7,1,3]]_2$ \cite{steane}. An LDI form, with parameters $[[7,1,\geq 3]]_q$, for this code was found in \cite{gunderman2020local}:

\begin{equation}
\phi_\infty(H_3)=
\setcounter{MaxMatrixCols}{15}
\begin{bmatrix}
 1 & 0 & 0 & 0 & 0 & 0 & 1 & | & 0 & 0 & 0 & 0 & 1 & 1 & 0\\
 0 & 1 & 0 & 0 & 0 & 0 & 0 & | & 0 & 0 & 0 & 1 & 1 & 1 & 0\\
 0 & 0 & 1 & 0 & 0 & 0 & 1 & | & 0 & 0 & 0 & 1 & 0 & 1 & 0\\
 0 & 0 & 0 & 1 & 0 & 0 & 0 & | & -1 & 1 & 0 & 0 & 0 & 0 & 1\\
 0 & 0 & 0 & 0 & 1 & 0 & 0 & | & 0 & 1 & -1 & 0 & 0 & 0 & 1\\
 0 & 0 & 0 & 0 & 0 & 1 & 0 & | & 1 & 1 & 1 & 0 & 0 & 0 & 0\\
\end{bmatrix}.
\end{equation}
We return this code to CSS form by applying DFTs on particles $4$, $5$, and $6$:
\begin{equation}
\phi_\infty(H_3)=
\setcounter{MaxMatrixCols}{15}
\begin{bmatrix}
 1 & 0 & 0 & 0 & 1 & 1 & 1 & | & 0 & 0 & 0 & 0 & 0 & 0 & 0\\
 0 & 1 & 0 & 1 & 1 & 1 & 0 & | & 0 & 0 & 0 & 0 & 0 & 0 & 0\\
 0 & 0 & 1 & 1 & 0 & 1 & 1 & | & 0 & 0 & 0 & 0 & 0 & 0 & 0\\
 0 & 0 & 0 & 0 & 0 & 0 & 0 & | & -1 & 1 & 0 & -1 & 0 & 0 & 1\\
 0 & 0 & 0 & 0 & 0 & 0 & 0 & | & 0 & 1 & -1 & 0 & -1 & 0 & 1\\
 0 & 0 & 0 & 0 & 0 & 0 & 0 & | & 1 & 1 & 1 & 0 & 0 & -1 & 0\\
\end{bmatrix}.
\end{equation}

Using the bound for $p^*$ without knowing the code is CSS but knowing that the maximal entry is $1$, we find $p^*=16$. This leaves a handful of primes still to verify the distance over. If instead we use Theorem \ref{css} with the known maximal entry of $1$ we obtain $p^*=1^22^{2/2}$, which is $2$. This means that the distance of this code is always at least preserved. While this was shown in that prior work, this provides a quicker way to obtain this fact and illustrates the improvements obtained here.

The prior description of an LDI representation for the Steane code used the prescriptive method for generating a local-dimension-invariant form for the code, however, we could have equivalently performed the following alteration:
\begin{equation}
\phi_2 (H_3)=
\setcounter{MaxMatrixCols}{15}
\begin{bmatrix}
 1 & 1 & 1 & 1 & 0 & 0 & 0 & | & 0 & 0 & 0 & 0 & 0 & 0 & 0\\
 0 & 1 & 1 & 0 & 1 & 1 & 0 & | & 0 & 0 & 0 & 0 & 0 & 0 & 0\\
 0 & 1 & 0 & 1 & 1 & 0 & 1 & | & 0 & 0 & 0 & 0 & 0 & 0 & 0\\
 0 & 0 & 0 & 0 & 0 & 0 & 0 & | & 1 & 1 & 1 & 1 & 0 & 0 & 0\\
 0 & 0 & 0 & 0 & 0 & 0 & 0 & | & 0 & 1 & 1 & 0 & 1 & 1 & 0\\
 0 & 0 & 0 & 0 & 0 & 0 & 0 & | & 0 & 1 & 0 & 1 & 1 & 0 & 1\\
\end{bmatrix}.
\end{equation}
We then put the code into an LDI form by flipping the signs of some of the entries, producing:
\begin{equation}\label{ldisteane}
\phi_\infty(H_3)=
\setcounter{MaxMatrixCols}{15}
\begin{bmatrix}
 1 & 1 & 1 & 1 & 0 & 0 & 0 & | & 0 & 0 & 0 & 0 & 0 & 0 & 0\\
 0 & 1 & 1 & 0 & 1 & 1 & 0 & | & 0 & 0 & 0 & 0 & 0 & 0 & 0\\
 0 & 0 & 1 & 1 & 0 & 1 & 1 & | & 0 & 0 & 0 & 0 & 0 & 0 & 0\\
 0 & 0 & 0 & 0 & 0 & 0 & 0 & | & 1 & -1 & 1 & -1 & 0 & 0 & 0\\
 0 & 0 & 0 & 0 & 0 & 0 & 0 & | & 0 & 1 & -1 & 0 & 1 & -1 & 0\\
 0 & 0 & 0 & 0 & 0 & 0 & 0 & | & 0 & 0 & 1 & -1 & 0 & -1 & 1\\
\end{bmatrix}.
\end{equation}
This also satisfies the same distance promise--always having distance at least 3, regardless of the local-dimension. In fact, we can extend this method of flipping the signs of some entries in this code to that of the whole family of codes within the $[[2^N-1,2^N-1-2N,3]]_2$ family. We show that this sign flipping can always generate an LDI representation for the code. In particular, we find that for this family we can prove a tight bound on the value of the maximal entry:

\begin{lemma}
For the qubit quantum Hamming code family with parameters $[[2^N-1,2^N-1-2N,3]]_2$, there is an LDI representation such that $B= 1$ for all members of the family.
\end{lemma}

Recall that this family is generated by each column being one of the nonzero binary strings of length $N$ in each the $X$ component and the $Z$ component. We will take the register placement of each string to be the same in the $X$ component and the $Z$ component.

\begin{proof} We begin by noting that for the $N=3$ case we already have the Steane code discussed just above (equation (\ref{ldisteane})). We will now prove inductively that we may always generate an LDI form such that $B=1$. Let $H_N^\infty$ be the $N$-th family member such that all generators in $H_N^\infty$ commute with those of $H_N$. Next, consider for $N\geq 4$, one can write the next member of the family in terms of the prior member as:

\begin{equation}
    \phi_2(H_{N+1})=\begin{bmatrix}
     1 & 1^{\otimes 2^N-1} & 0^{\otimes 2^N-1} & | & 0 & 0 & 0\\
     0 & H_N & H_N & | & 0 & 0 & 0\\
     0 & 0 & 0 & | & 1 & 1^{\otimes 2^N-1} & 0^{\otimes 2^N-1}\\
     0 & 0 & 0 & | & 0 & H_N^\infty & H_N^\infty
    \end{bmatrix},
\end{equation}
with the superscript $\otimes$ indicating repetition of that value. However, this version of the code is not in an LDI representation. We make the following sign changes and then verify that this version is an LDI representation:
\begin{equation}
    \phi_\infty(H_{N+1}^\infty)=\begin{bmatrix}
     1 & 1^{\otimes 2^N-1} & 0^{\otimes 2^N-1} & | & 0 & 0 & 0\\
     0 & H_N & H_N & | & 0 & 0 & 0\\
     0 & 0 & 0 & | & 1 & (-1\ 1)^{\otimes 2^{N-1}-2}\bigoplus (-1) & 0^{\otimes 2^N-1}\\
     0 & 0 & 0 & | & 0 & H_N^\infty & H_N^\infty
    \end{bmatrix},
\end{equation}
where $\bigoplus$ indicates a direct sum, tacking that value onto the end of the repetition. By the inductive hypothesis all operators in $H_N$ commute with those of $H_N^\infty$. Next, the first row commutes with the first nontrivial row in the $Z$ component as there are an equal number of $-1$ and $+1$. The sum of each row in $H_N^\infty$ is $0$, again due to alternating signs, and so the first row commutes with all rows in $H_N^\infty$. We must lastly ensure that $(-1\ 1)^{\otimes 2^{N-1}-2}\bigoplus (-1)$ commutes with each row in $H_N$. We will denote $v=(-1\ 1)^{\otimes 2^{N-1}-2}\bigoplus (-1)$. Consider the most recently added row in $H_N$. This will be given by $1^{\otimes 2^{N-1}}\bigoplus 0^{\otimes 2^{N-1}-1}$, which will commute with $v$. The following rows will be those of $(0\bigoplus H_{N-2}\bigoplus H_{N-2})^{\otimes 2}$, for which each row commutes with $v$ as for each register in $H_{N-2}$ there will be one time where it is added, $+1$ in $v$, and one time where it will be subtracted, $-1$ in $v$.
\end{proof}
\if{false}
\begin{proof}
For this proof it will suffice to show that the dot product between any pair of rows in just the $\mathcal{X}$ component can be made $0$ while only altering the signs of the entries. This is equivalent to find a local-dimension-invariant version of the classical Hamming code family. We will prove this in parts, using induction. We denote by $H_n^\infty$ the LDI version of the $n$-th family member of the Hamming family.

For the base case, we require two initial LDI representations for the code, which can be given by:

\setcounter{MaxMatrixCols}{24}
\begin{equation}
    H_3^\infty=\begin{bmatrix}
        -1 & 1 & 1 & 1 & 0 & 0 & 0\\
        1 & 1 & 0 & 0 & -1 & 1 & 0\\
        1 & 0 & 1 & 0 & 1 & 0 & 1\\
    \end{bmatrix}
\end{equation}

\begin{equation}
    H_4^\infty=\begin{bmatrix}
        1 & 1 & 1 & -1 & 1 & -1 & -1 & 1 & 0 & 0 & 0 & 0 & 0 & 0 & 0\\
        -1 & 1 & 1 & 1 & 0 & 0 & 0 & 0 & -1 & 1 & 1 & 1 & 0 & 0 & 0\\
        1 & 1 & 0 & 0 & -1 & 1 & 0 & 0 & 1 & 1 & 0 & 0 & -1 & 1 & 0\\
        1 & 0 & 1 & 0 & 1 & 0 & 1 & 0 & 1 & 0 & 1 & 0 & 1 & 0 & 1\\
    \end{bmatrix}
\end{equation}

Following this we take as our inductive hypothesis that there is a procedure to generate $H_n^\infty$ and $H_{n-1}^\infty$, and will induct for the case of $H_{n+1}^\infty$. In which case by the structure of the Hamming code we have:
\begin{equation}
    H_{n+1}=\begin{bmatrix}
        1^{\otimes 2^{n-1}-1} & 1 & 1^{\otimes 2^{n-1}-1} & 1 & 0^{\otimes 2^{n-1}-1} & 0 & 0^{\otimes 2^{n-1}-1}\\
        1^{\otimes 2^{n-1}-1} & 1 & 0^{\otimes 2^{n-1}-1} & 0 & 1^{\otimes 2^{n-1}-1} & 1 & 0^{\otimes 2^{n-1}-1}\\
        H_{n-1}^\infty & \mathbf{0} & H_{n-1}^\infty & \mathbf{0} & H_{n-1}^\infty & \mathbf{0} & H_{n-1}^\infty
    \end{bmatrix}.
\end{equation}
In this form the last $n-1$ rows all have dot products of zero, due to the inductive hypothesis. To ensure the second row also has dot product zero with all of the rows beneath it, we may make the following sign changes:
\begin{equation}
    H_{n+1}=\begin{bmatrix}
        1^{\otimes 2^{n-1}-1} & 1 & 1^{\otimes 2^{n-1}-1} & 1 & 0^{\otimes 2^{n-1}-1} & 0 & 0^{\otimes 2^{n-1}-1}\\
        1^{\otimes 2^{n-1}-1} & 1 & 0^{\otimes 2^{n-1}-1} & 0 & (-1)^{\otimes 2^{n-1}-1} & 1 & 0^{\otimes 2^{n-1}-1}\\
        H_{n-1}^\infty & \mathbf{0} & H_{n-1}^\infty & \mathbf{0} & H_{n-1}^\infty & \mathbf{0} & H_{n-1}^\infty
    \end{bmatrix}.
\end{equation}
The final challenge is to ensure the first row also has zero dot product with the other rows. The last $n-1$ rows will have a zero dot product so long as the two repetitions of $1$'s in the first row have complementary signs--meaning that for each position that there is a $+1$ the same position in the other repetition has a $-1$. We select as our choice the repetitions: $v:=(1\bigoplus -1)^{\otimes 2^{n-2}}\bigoplus 1$ and $\bar{v}:=(-1\bigoplus 1)^{\otimes 2^{n-2}}\bigoplus (-1)$ and set the middle $1$ to $-1$, producing
\begin{equation}
    H_{n+1}=\begin{bmatrix}
        v & -1 & \bar{v} & 1 & 0^{\otimes 2^{n-1}-1} & 0 & 0^{\otimes 2^{n-1}-1}\\
        1^{\otimes 2^{n-1}-1} & 1 & 0^{\otimes 2^{n-1}-1} & 0 & (-1)^{\otimes 2^{n-1}-1} & 1 & 0^{\otimes 2^{n-1}-1}\\
        H_{n-1}^\infty & \mathbf{0} & H_{n-1}^\infty & \mathbf{0} & H_{n-1}^\infty & \mathbf{0} & H_{n-1}^\infty
    \end{bmatrix}.
\end{equation}
This has zero dot product with the last $n-1$ rows since for every $+1$ there is a $-1$ for each register in $H_{n-1}^\infty$. Lastly the first two rows have zero dot product as exactly half of $v\bigoplus (-1)$ is $-1$ and half is $+1$, while the second row is all $+1$ for those registers. This completes the induction. [check values, but otherwise seems good--self innerproduct is not zero, instead use a form like: $[H_n\ 0\ 0\ H_n^\infty]$, with $H_n^\infty$ having alternating signs]
\end{proof}\fi

While the ability to write this family in an LDI representation only using $\{0,\pm 1\}$ as the entries is of limited interest, applying this result with Theorem \ref{css} we obtain:

\begin{corollary}
All qubit Hamming codes have an LDI representation that has distance at least $3$, meaning that this generates a $[[2^N-1,2^N-2N-1,\geq 3]]_q$ code family for all $N\geq 3$ and $q$ a prime.
\end{corollary}

\begin{proof}
This family is a non-degenerate CSS code family. This result then follows from the CSS distance promise for LDI codes and the above Lemma showing that $B=1$ may be achieved for this family.
\end{proof}

This shows that the local-dimension-invariant form is at least in some cases able to provide tight expressions that allow for the full importing of code families for larger local-dimension systems. Additionally, this provides another qudit code family with two parameters $N$ and $q$.

%The ratio as a function of $N$ is shown in Figure 3, while Figure 4 shows the reduction of the region of uncertainty. Unfortunately, this still leaves some constant fraction of the region, a region which itself grows exponentially. The following result combines this Lemma with that of the reduction for $p^*$ when the code is CSS to allow this fraction to tend toward zero:

\if{false}
\begin{figure}[htb]
\includegraphics[width=\textwidth]{Website Screenshot.png}
\caption{This is a screenshot of the web page which allows others to utilize my program. It is located at \url{bit.ly/qGenerator}.}
\centering
\end{figure}
\fi

\section{Discussion}

While few qudit quantum computer prototypes are currently being built this provides another avenue for expanding computational power. \if{false}for following a quantum Moore's law and may at times be a more practical method for expanding the power. Being able to create these customized codes will be of use while constructing small qudit quantum computers, however, once larger computers are being built this method will likely be impractical so local-dimension-invariant codes will be of more use.\fi While currently the size of the local-dimension must be large to promise the distance of an arbitrary LDI code, we have provided a case where this can be at least quadratically improved. \if{false}In addition, having the tools provided here to generate more examples of LDI codes will allow one to test further hypotheses allowing for this cutoff value to be further reduced through theoretical techniques. These techniques can be used in conjunction to form new codes, even with parameters forbidden by codes with a smaller local-dimension, then can be used to show that the distance can be promised even if the local-dimension is increased.\fi Beyond this improvement we have also constructed a new qudit code family that is directly imported from a qubit code family. While there already exists qudit quantum Hamming families with paramters $[[n,n-2m,3]]_q$ for $m\geq 2$ in the cases of $gcd(m,q^2-1)=1,\ n=(q^{2m}-1)/(q^2-1)$ and $gcd(m,q-1)=1,\ n=(q^m-1)/(q-1)$, this new code family fills in values of $n$ which are not covered by these \cite{ketkar2006nonbinary}. Additionally, while there exists $[[n,n-4,3]]_q$ codes for odd prime power $q$ values and $4\leq n\leq q^2+1$, this provides options for $n$ beyond $q^2+1$ \cite{li2010construction}. Lastly, while arguably Maximally-Distance-Separable (MDS) codes are optimal, there are a number of values of $n$ for which there are no known MDS codes for a given value of $q$ \cite{shi2019new}. While the family presented here is not MDS, perhaps the analysis used in this work can be applied to help fill in missing parameter choices. This work has only concerned itself with preserving the distance of codes in LDI representations, but investigating whether other desirable properties are preserved is also an important direction.  % A website for performing these operations is available at \url{bit.ly/qGenerator}, as well as a link to the github for the source code\footnote{\url{https://github.com/arunjmoorthy/qGeneratorCode}}. 

\if{false}
\section*{Acknowledgments}

We thank Elijah Durso-Sabina for the suggestion of implementing equation \ref{wts} and pointing us to the use of such for flag fault-tolerant error-correction.
\fi
\section*{Funding}
L.G. was supported by Industry Canada, the Canada
First Research Excellence Fund (CFREF), the Canadian Excellence Research Chairs (CERC 215284) program, the Natural Sciences and Engineering Research Council of Canada
(NSERC RGPIN-418579) Discovery program, the Canadian
Institute for Advanced Research (CIFAR), and the Province
of Ontario.

\if{false}
\section{Introduction}        
The journal of {\it Quantum Information and Computation},
for both on-line and in-print editions,
will be produced by using the latex files of manuscripts
provided by the authors. It is therefore essential that the manuscript 
be in its final form, and in the format designed for the journal 
because there will be no futher editing. The authors are strongly encouraged 
to use Rinton latex template to prepare their manuscript. Or, the authors 
should please follow the instructions given here if they prefer to use other 
software. In the latter case, the authors ought to
provide a postscript file of their paper for publication.

\section{Text}
\noindent
Contributions are to be in English. Authors are encouraged to
have their contribution checked for grammar.  
Abbreviations are allowed but should be spelt
out in full when first used. 

\setcounter{footnote}{0}
\renewcommand{\thefootnote}{\alph{footnote}}

The text is to be typeset in 10 pt Times Roman, single spaced
with baselineskip of 13 pt. Text area (excluding running title)
is 5.6 inches across and 8.0 inches deep.
Final pagination and insertion of running titles will be done by
the editorial. Number each page of the manuscript lightly at the
bottom with a blue pencil. Reading copies of the paper can be
numbered using any legible means (typewritten or handwritten).

\section{Headings}
\noindent
Major headings should be typeset in boldface with the first
letter of important words capitalized.

\subsection{Sub-headings}
\noindent
Sub-headings should be typeset in boldface italic and capitalize
the first letter of the first word only. Section number to be in
boldface roman.

\subsubsection{Sub-subheadings}
\noindent
Typeset sub-subheadings in medium face italic and capitalize the
first letter of the first word only. Section number to be in
roman.

\subsection{Numbering and Spacing}
\noindent
Sections, sub-sections and sub-subsections are numbered in
Arabic.  Use double spacing before all section headings, and
single spacing after section headings. Flush left all paragraphs
that follow after section headings.

\subsection{Lists of items}
\noindent
Lists may be laid out with each item marked by a dot:
\begin{itemlist}
 \item item one,
 \item item two.
\end{itemlist}
Items may also be numbered in lowercase roman numerals:
\begin{romanlist}
 \item item one
 \item item two
          \begin{alphlist}
          \item Lists within lists can be numbered with lowercase
              roman letters,
          \item second item.
          \end{alphlist}
\end{romanlist}

\section{Equations}
\noindent
Displayed equations should be numbered consecutively in each
section, with the number set flush right and enclosed in
parentheses.

\begin{equation}
\mu(n, t) = {
\sum^\infty_{i=1} 1(d_i < t, N(d_i) = n) \over \int^t_{\sigma=0} 1(N(\sigma) 
= n)d\sigma}\,. \label{this}
\end{equation}

Equations should be referred to in abbreviated form,
e.g.~``Eq.~(\ref{this})'' or ``(2)''. In multiple-line
equations, the number should be given on the last line.

Displayed equations are to be centered on the page width.
Standard English letters like x are to appear as $x$
(italicized) in the text if they are used as mathematical
symbols. Punctuation marks are used at the end of equations as
if they appeared directly in the text.

\vspace*{12pt}
\noindent
{\bf Theorem~1:} Theorems, lemmas, etc. are to be numbered
consecutively in the paper. Use double spacing before and after
theorems, lemmas, etc.

\vspace*{12pt}
\noindent
{\bf Proof:} Proofs should end with \square\,.

\section{Illustrations and Photographs}
\noindent
Figures are to be inserted in the text nearest their first
reference. The postscript files of figures can be imported by using
the commends used in the examples here.

\begin{figure} [htbp]
\vspace*{13pt}
\centerline{\epsfig{file=fig1-eps-converted-to.pdf, width=8.2cm}} %100 percent
\vspace*{13pt}
\fcaption{\label{motion}figure caption goes here.}
\end{figure}

Figures are to be sequentially numbered in Arabic numerals. The
caption must be placed below the figure. Typeset in 8 pt Times
Roman with baselineskip of 10~pt. Use double spacing between a
caption and the text that follows immediately.

Previously published material must be accompanied by written
permission from the author and publisher.

\section{Tables}
\noindent
Tables should be inserted in the text as close to the point of
reference as possible. Some space should be left above and below
the table.

Tables should be numbered sequentially in the text in Arabic
numerals. Captions are to be centralized above the tables.
Typeset tables and captions in 8 pt Times Roman with
baselineskip of 10 pt.

\vspace*{4pt}   %only when needed
\begin{table}[hb]
\tcaption{Number of tests for WFF triple NA = 5, or NA = 8.}
\centerline{\footnotesize NP}
\centerline{\footnotesize\smalllineskip
\begin{tabular}{l c c c c c}\\
\hline
{} &{} &3 &4 &8 &10\\
\hline
{} &\phantom03 &1200 &2000 &\phantom02500 &\phantom03000\\
NC &\phantom05 &2000 &2200 &\phantom02700 &\phantom03400\\
{} &\phantom08 &2500 &2700 &16000 &22000\\
{} &10 &3000 &3400 &22000 &28000\\
\hline\\
\end{tabular}}
\end{table}

If tables need to extend over to a second page, the continuation
of the table should be preceded by a caption, e.g.~``({\it Table
2. Continued}).''

\section{References Cross-citation}
\noindent
References cross-cited in the text are to be numbered consecutively in
Arabic numerals, in the order of first appearance. They are to
be typed in brackets such as \cite{first}  and \cite{cal, niel, mar}.
%superscripts after punctuation marks,
%e.g.~``$\ldots$ in the statement.$^5$''.

\section{Sections Cross-citation}\label{sec:abc}
\noindent
Sections and subsctions can be cross-cited in the text by using the latex command
shown here. In Section~\ref{sec:abc}, we discuss ....
%\newpage

\section{Footnotes}
\noindent
Footnotes should be numbered sequentially in superscript
lowercase Roman letters.\fnm{a}\fnt{a}{Footnotes should be
typeset in 8 pt Times Roman at the bottom of the page.}

\nonumsection{Acknowledgements}
\noindent
We would thank ...
\fi

%\nonumsection{References}
%\noindent
%References are to be listed in the order cited in the text.
%For each cited work, include all the authors' names, year of the work, title,
%place where the work appears.
%Use the style shown in the following examples. For journal names,
%use the standard abbreviations. Typeset references in 9 pt Times
%Roman.

\if{false}

\fi
%\bibliography{ieeetr}
%\phantomsection  
%\renewcommand*{\bibname}{References}
%\bibliography{main}

\bibliographystyle{unsrt}
\phantomsection  
\renewcommand*{\bibname}{References}

\bibliography{main}

%\section*{Appendix}

\if{false}
\vspace{0.17cm}
{\bf Generators for the prior codes so that the generators satisfy $wt_p(\mathcal{C})$:}
\setcounter{MaxMatrixCols}{24}

\begin{equation}
[[6,2,3]]_3=\begin{bmatrix}
0 & 1 & 1 & 1 & 0 & 2 & | & 0 & 0 & 2 & 0 & 1 & 0\\
0 & 1 & 1 & 2 & 1 & 0 & | & 2 & 1 & 0 & 2 & 0 & 0\\
1 & 0 & 1 & 2 & 2 & 1 & | & 2 & 0 & 1 & 2 & 1 & 2\\
2 & 2 & 2 & 1 & 0 & 1 & | & 1 & 1 & 0 & 2 & 0 & 2
\end{bmatrix}
\end{equation}

\begin{equation}
[[7,3,3]]_3=\begin{bmatrix}
1 & 1 & 0 & 1 & 2 & 2 & 2 & | & 0 & 0 & 0 & 0 & 0 & 2 & 0\\
2 & 0 & 1 & 2 & 1 & 0 & 1 & | & 0 & 2 & 0 & 2 & 2 & 0 & 2\\
2 & 1 & 2 & 0 & 1 & 1 & 0 & | & 1 & 0 & 1 & 0 & 2 & 1 & 2\\
2 & 0 & 1 & 0 & 2 & 2 & 2 & | & 0 & 0 & 1 & 2 & 1 & 1 & 2
\end{bmatrix}
\end{equation}

\begin{equation}
[[8,4,3]]_3=\begin{bmatrix}
1 & 0 & 0 & 1 & 0 & 2 & 1 & 1 & | & 1 & 0 & 2 & 2 & 2 & 0 & 0 & 2\\
0 & 1 & 2 & 1 & 2 & 1 & 0 & 2 & | & 0 & 0 & 1 & 1 & 2 & 0 & 1 & 0\\
1 & 0 & 1 & 1 & 0 & 1 & 2 & 0 & | & 1 & 1 & 2 & 1 & 0 & 2 & 1 & 1\\
1 & 1 & 2 & 0 & 0 & 1 & 1 & 1 & | & 2 & 1 & 0 & 2 & 0 & 0 & 0 & 0
\end{bmatrix}
\end{equation}

\begin{equation}
[[10,5,3]]_3=\begin{bmatrix}
1 & 1 & 0 & 2 & 0 & 0 & 2 & 1 & 0 & 2 & | & 2 & 2 & 0 & 1 & 0 & 0 & 1 & 2 & 0 & 0\\
0 & 0 & 1 & 0 & 2 & 0 & 0 & 2 & 0 & 0 & | & 0 & 1 & 0 & 2 & 0 & 2 & 2 & 1 & 0 & 0\\
1 & 1 & 0 & 0 & 2 & 1 & 0 & 1 & 2 & 1 & | & 1 & 2 & 0 & 0 & 1 & 1 & 0 & 1 & 1 & 1\\
2 & 2 & 1 & 0 & 0 & 2 & 1 & 0 & 2 & 0 & | & 2 & 0 & 0 & 0 & 1 & 0 & 2 & 0 & 0 & 0\\
0 & 2 & 0 & 0 & 1 & 1 & 2 & 0 & 0 & 2 & | & 1 & 1 & 1 & 0 & 1 & 2 & 2 & 1 & 0 & 1
\end{bmatrix}
\end{equation}

\if{false}
\appendix

\noindent
Appendices should be used only when absolutely necessary. They
should come after the References. If there is more than one
appendix, number them alphabetically. Number displayed equations
occurring in the Appendix in this way, e.g.~(\ref{that}), (A.2),
etc.
\begin{equation}
\langle\hat{O}\rangle=\int\psi^*(x)O(x)\psi(x)d^3x~. 
\label{that}
\end{equation}
\fi
\fi

\end{document}